\documentclass[11pt]{article}
\usepackage{tipa}
\usepackage{geometry}                
\usepackage{amssymb,amsmath,amsthm}
\usepackage{epstopdf}
\usepackage{wrapfig}
\usepackage{graphicx}
\usepackage{arydshln}
\usepackage{setspace}

\usepackage{mathtools}
\usepackage{blkarray, bigstrut}

\usepackage{marvosym}
\usepackage{tikz}
\usepackage{verbatim}

\newcommand{\CC}{\mathbb C}

\newcommand{\ii}{{\rm i}}
\newcommand{\ee}{{\rm e}}
\newcommand{\dd}{{\rm d}}

\newtheorem{lemma}{Lemma}
\newtheorem{definition}{Definition}
\newtheorem{cor}{Corollary}
\newtheorem{theorem}{Theorem}
\newtheorem{proposition}{Proposition}
\newtheorem*{thm-others}{Theorem}

\newtheoremstyle{thm-others}
  {}
  {}
  {}
  {}
  {bold}
  {}
  {}{}

\title{
{Planar Orthogonal Polynomials\\ As Type II Multiple Orthogonal
Polynomials}}

\author{Seung-Yeop Lee, Meng Yang}

\date{}
\begin{document}
\maketitle

\abstract{We show that the planar orthogonal polynomials with $l$
logarithmic singularities in the potential are the multiple
orthogonal polynomials (Hermite-Pad\'e polynomials) of Type II with
$l$ measures. We also find the ratio between the determinant of the
moment matrix corresponding to the multiple orthogonal polynomials
and the determinant of the moment matrix from the original planar
measure.}

\section{Main Result}

Let $p_{n}(z)$ be the monic polynomial of degree $n$ satisfying the
orthogonality:
\begin{equation}\label{eq1}
\int_\CC p_n(z)\,\overline{p_m(z)}\,\ee^{-|z|^2} |W(z)|^2\,\dd
A(z)=h_n\delta_{nm},\quad n,m\geq 0,\end{equation} where $\dd A$ is
the Lebesgue area measure of the complex plane and $h_n$ is the
positive norming constant. We define, for $l\geq 1,$ the
multi-valued function $W$ by
\begin{equation}\label{eqw}
    W(z)=\prod_{j=1}^l(z-a_j)^{c_j } ,\quad z\in\CC,
\end{equation}
where $\{c_1,\cdots,c_l\}$ are positive real numbers and
$\{a_1,\cdots,a_l\}$ are distinct points in $\CC.$

The orthogonal polynomial whose measure is supported on the plane is
called {\it{planar orthogonal polynomial}}. Such polynomial has been
of interest due to its connection to two--dimensional Coulomb gas
\cite{Ma 2011}. Moreover the polynomial defined above appears
\cite{Ra 2005} in the quantized version of Hele-Shaw flow, a type of
growth model in the two--dimensional plane. These connections to
physical system, Coulomb gas and Hele-Shaw flow, motivate one to
study the large degree behavior of the polynomials. We recommend the
recent paper \cite{Ha 2017} for an important progress in this regard
and for the related history. Still lacking, until now, is the
understanding of the limiting zero distribution when the degree of
the polynomial goes to infinity. Several case studies \cite{Ba1
2017, ku94 2015, Ba 2015, ku104 2015, ku103 2015, ku106 2015} have
shown that the zeros tend to certain one--dimensional set. In all of
these cases the planar orthogonal polynomial in question turns out
to be also either a classical orthogonal polynomial or a multiple
orthogonal polynomial \cite{Fi 2017, Ku 2010}, of which the
asymptotic behavior is possible to study \cite{Ra 2017} due to rich
algebraic structure such as finite term recurrence relation.

The main result of the paper is that our polynomials $\{p_n\}$ are
multiple orthogonal polynomials of Type II. To introduce the main
theorem, let us prepare several notations. To remove the unnecessary
complication, we assume that $a_j$'s are all nonzero and the
arguments of $a_j$'s are all different. Without loss of generality,
we may assume:
\begin{equation}
    0\leq \arg a_1<\cdots < \arg a_l< 2\pi.
\end{equation}
To determine the branch of the multi-valued function $W$, we define
the union of contours,
\begin{equation}
    {\bf B} = \bigcup_{j=1}^l {\bf B}_j,  \quad {\bf B}_j=
    \{a_j\,t:\,t\geq 1\},
\end{equation}
where the contours are directed towards the infinity. In the rest of
the paper, we define $W: \CC\setminus {\bf B}\to \CC$ be an analytic
branch of \eqref{eqw}. Let ${\bf B}^*$ and ${{\bf B}}^*_j$ be the
complex-conjugate images of ${\bf B}$ and ${\bf B}_j$. Let
$\overline{W}:\CC\setminus{\bf B}^*\to\CC$ be defined by
\begin{equation}
   \displaystyle \overline {W}(z)= \overline{W(\bar{z})} =
    \prod_{j=1}^l(z-\bar{a}_j)^{c_j}.
\end{equation}
Let ${\bf{k}}=(k_1,\cdots,k_l)$ with non--negative integers $k_j$'s.
When $\arg z\notin \{\arg a_1,\cdots,\arg a_l\}$, we
 define
\begin{equation}\label{eq47}
\chi_{\bf{k}}(z)=
W(z)\int_{0}^{\bar{z}\times\infty}\prod_{j=1}^l(s-\bar{a}_j)^{k_j}\overline{W}(s)e^{-zs}\,\dd
s,
\end{equation}
where the represented integration contour is $\{\bar{z}t|\, t\geq
0\}$.
\begin{definition}
Let $\Gamma$ be a simple closed curve with counterclockwise
orientation, that connects $\{a_1,\cdots,a_l\}$, encloses the
origin, and does not intersect ${\bf B}\setminus\{a_1,\cdots,a_l\}$.
Explicitly, we may choose
$\Gamma=\overline{a_1a_2}\cup\cdots\cup\overline{a_{l-1}a_l}\cup\overline{a_la_1},$
by the union of $l$ line segments.
\end{definition}

\begin{figure*}
\begin{center}
\includegraphics[width=0.45\textwidth]{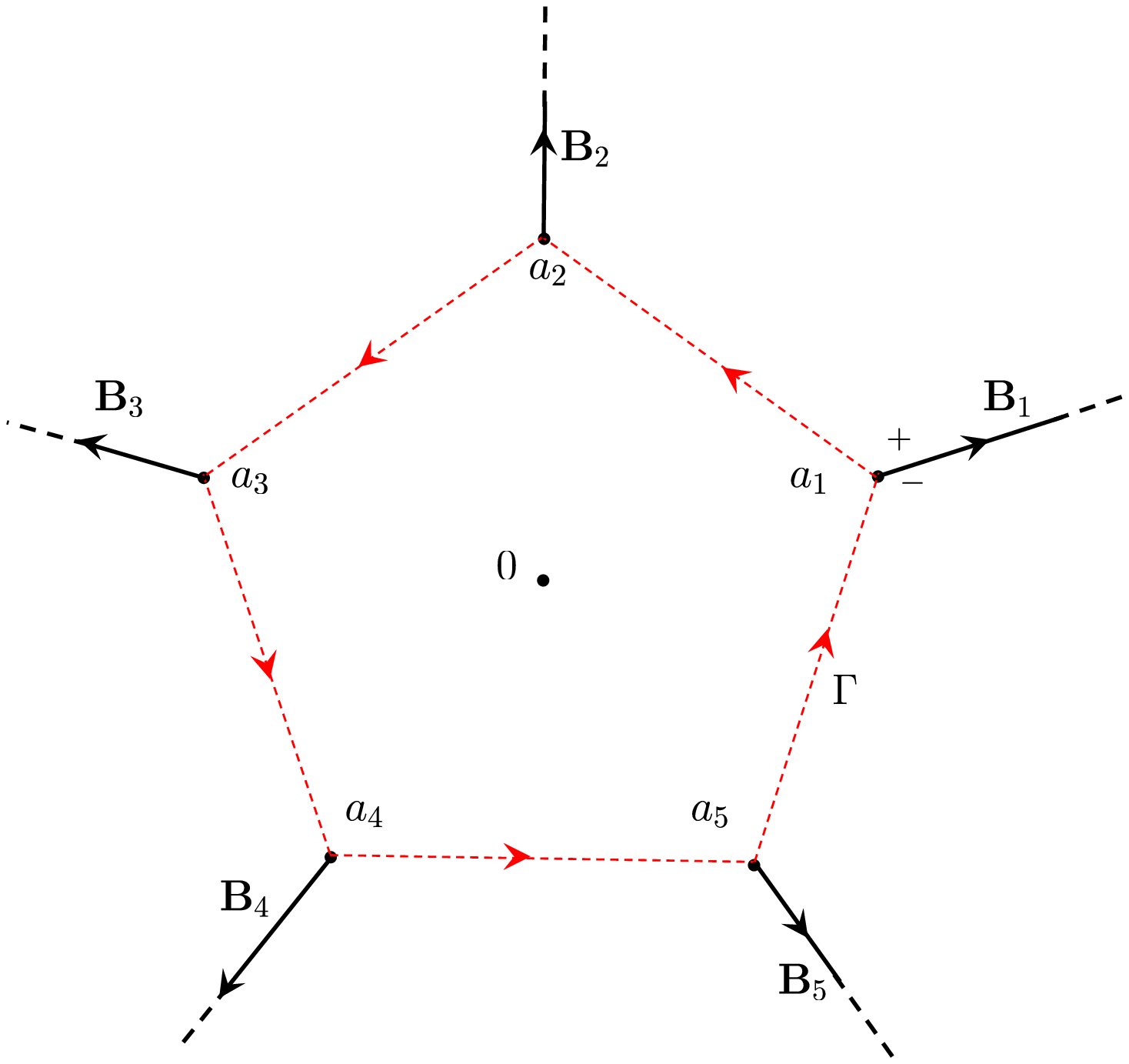}
\qquad
\includegraphics[width=0.45\textwidth]{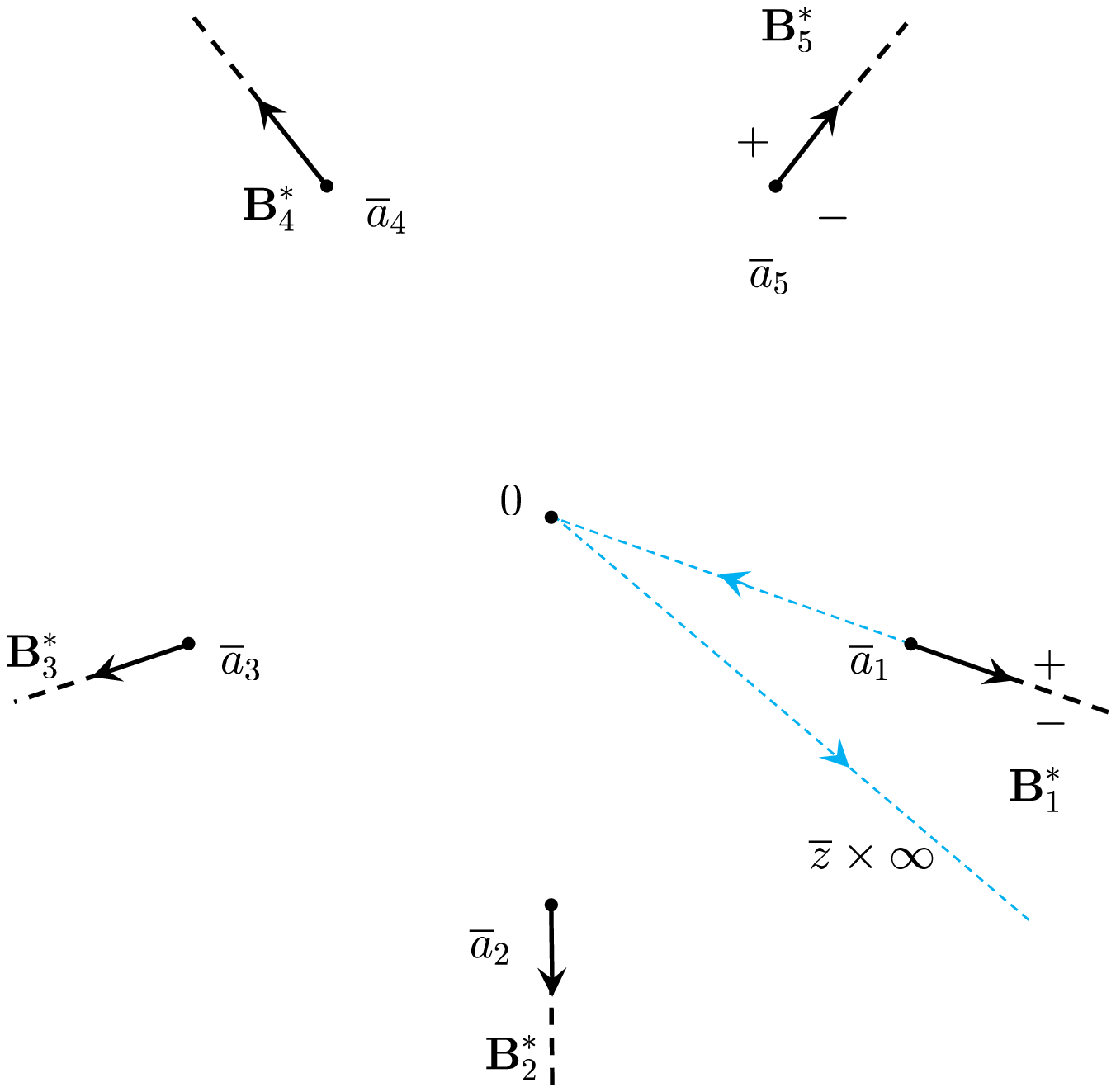}
\caption{\label{fig-branch-cuts} Contours when $l=5$.  In the left are contours for ${\bf B}$ (black) and $\Gamma$ (dotted red); In the right are the complex conjugate image of the right, and the integration contour for $\widetilde{\chi}_{\bf k}$ (dotted blue).}
\end{center}
\end{figure*}

\begin{definition}
Let ${\bf{n}}=(n_1,\cdots,n_l)$ with non--negative integers $n_j$'s.
We define $p_{{\bf{n}}}(z)$ to be the monic polynomial of degree
$|{\bf{n}}|=n$ satisfying the orthogonality condition:
\begin{equation}\label{eq52}
\displaystyle\int_\Gamma
p_{{\bf{n}}}(z)\,z^k\chi_{{\bf{n}}-{\bf{e}}_j}(z)\,\dd z=0, \quad
0\leq k\leq n_j-1,\quad 1\leq j\leq l.\end{equation} Here
${\bf{e}}_j$ is the unit vector with one at the $j$th entry and
zeros at all the other entries. We define $q_{\bf n}^{(i)}(z)$ to be
the monic polynomial of degree $|{\bf{n}}|-1$ satisfying the
orthogonality condition:
\begin{equation}\nonumber
\displaystyle \int_\Gamma q_{\bf
n}^{(i)}(z)\,z^k\chi_{{\bf{n}}-{\bf{e}}_j}(z)\,\dd z=0, \quad 0\leq
k\leq n_j-1-\delta_{ij},\quad 1\leq i,\, j\leq l.\end{equation} The
polynomials $p_{{\bf{n}}}(z)$ and $q_{\bf n}^{(i)}(z)$ are multiple
orthogonal polynomials of type II.
\end{definition}

Multiple orthogonal polynomials are related to Hermite--Pad\'e
approximation to a system of Markov functions \cite{Assche 2001}.
For type II Hermite--Pad\'e approximation, we look for rational
functions approximating Markov functions near infinity, which
consists of finding a polynomial $P_{\bf n}$ of degree $|{\bf n}|$
and polynomials $Q_{{\bf n},j}$ $(j=1,\cdots,l)$ of degree less than
$|{\bf n}|$ such that
\begin{equation}\nonumber P_{\bf n}(z)f_j(z)-Q_{{\bf n},j}(z)=\mathcal
{O}\left(\frac{1}{z^{n_j+1}}\right),\quad z\to\infty,\,\,
j=1,\cdots,l,\end{equation} where $f_1,\cdots,f_l$ are $l$ Markov
functions given, in our context, by
\begin{equation}\nonumber f_j(z)=\displaystyle\int_\Gamma\frac{\chi_{{\bf{n}}-{\bf
e}_j}(s)}{z-s}\,\dd s,\quad z\notin \Gamma,\,\,
j=1\cdots,l.\end{equation} Then $Q_{{\bf n},j}(z)$ is given by
\begin{equation}\nonumber\displaystyle Q_{{\bf n},j}(z)=\int_\Gamma\frac{\left(P_{\bf n}(z)-P_{\bf n}(s)\right)
\chi_{{\bf{n}}-{\bf e}_j}(s)}{z-s}\,\dd s.\end{equation} In our
context, $P_{\bf n}=p_{\bf n}.$ We now state the main results:

\begin{theorem}\label{thm1}
Given positive integers $n$ and $l,$ we define a non--negative
integer $\kappa$ and a non--negative integer $0\leq r<l$ such that
$n=\kappa l+r.$ Then,
$$p_n(z)=p_{\bf{n}}(z),$$ where
${\bf{n}}={\bf{n}}(n,l)=(\underbrace{\kappa+1,\cdots,\kappa+1}_{r},\underbrace{\kappa,\cdots,\kappa}_{l-r}).$
\end{theorem}

The next theorem is an immediate
consequence; see \cite{Ku 2010} for a reference.
\begin{theorem}
Let $n,\,l,\,\kappa,\,r$ and ${\bf n}$ be given as in Theorem
\ref{thm1}. Let the $(l+1)$ by $(l+1)$ matrix function $Y$ be given
by
\[
Y(z) =\begin{blockarray}{ccccc}
\begin{block}{[cccc]c}
p_{\bf{n}}(z) &\displaystyle \frac{1}{2\pi
\mathrm{i}}\int_\Gamma\frac{p_{\bf{n}}(w)\chi_{{\bf{n}}-{\bf
e}_1}(w)}{w-z}\,\dd w
 &\cdots &\displaystyle \frac{1}{2\pi \mathrm{i}}\int_\Gamma\frac{p_{\bf{n}}(w)\chi_{{\bf{n}}-{\bf e}_l}(w)}{w-z}\,\dd w&  \vspace{0.1cm}\\
\vdots&\vdots&\vdots&\vdots &  \vspace{0.1cm}\\
\gamma_j\,q_{\bf{n}}^{(j)}(z)&\displaystyle \frac{\gamma_j}{2\pi
\mathrm{i}} \int_\Gamma\frac{q_{\bf{n}}^{(j)}(w)\chi_{{\bf{n}}-{\bf
e}_1}(w)}{w-z}\,\dd w&\cdots&\displaystyle \frac{\gamma_j}{2\pi
\mathrm{i}} \int_\Gamma\frac{q_{\bf{n}}^{(j)}(w)\chi_{{\bf{n}}-{\bf
e}_l}(w)}{w-z}\,\dd w & \leftarrow(j+1)th\,\, \text{row}, \vspace{0.1cm}\\
\vdots&\vdots&\vdots&\vdots &  \vspace{0.1cm}\\
\end{block}
\end{blockarray}
 \] where the constant $\gamma_j$ in the $(j+1)$th row is given by
$$\displaystyle \gamma_{j}=-\left(\frac{1}{2\pi \mathrm{i}} \int_\Gamma
q_{\bf{n}}^{(j)}(w)w^m\chi_{{\bf{n}}-{\bf e}_j}(w)\,\dd
w\right)^{-1}, \qquad
     m = \begin{cases} \kappa \quad
     \mbox{for}\quad 1\leq j\leq r; \\ \kappa -1   \quad
     \mbox{for}\quad  r+1\leq j\leq l. \end{cases}  $$
Then the matrix function $Y$ is the unique solution to the
Riemann-Hilbert problem given below.
\begin{equation}\nonumber
\left\{\begin{array}{lll}
Y: \CC\setminus{\Gamma}\to\CC^{(l+1)\times (l+1)} \text{ is holomorphic matrix function}; \\
\\
Y_+(z)=Y_-(z)J(z) \mbox{ on $\Gamma$; } \\
\\\displaystyle
Y(z)=\left(I+\mathcal
{O}\left(\frac{1}{z}\right)\right)\begin{bmatrix}z^n&{\bf{0}}&{\bf{0}}\\{\bf{0}}&z^{-(\kappa+1)}I_{r\times r}&{\bf{0}}\\
{\bf{0}}&{\bf{0}}&z^{-\kappa}I_{(l-r)\times(l-r)}\end{bmatrix},\quad\mbox{as\,
$z\to\infty$}.
\end{array}\right.
\end{equation}
Above, the subscript $\pm$ in $Y_{\pm}$ represents the limiting value when approaching $\Gamma$ from the corresponding sides of the directed contour, and
\begin{equation}\nonumber J(z)=\begin{bmatrix}
1&\chi_{{\bf{n}}-{\bf e}_1}(z)&\cdots&\chi_{{\bf{n}}-{\bf e}_l}(z) \\
0&1&\cdots&0\\
\vdots&\vdots&\ddots&\vdots\\
0&0&\cdots&1
\end{bmatrix}.\end{equation}
\end{theorem}
\noindent{\bf{Remark.}} For $l=1$, the contour $\Gamma$ is a closed
curve around
 the origin passing through $a_1.$ After a little computation
one can see that the jump contour $\Gamma$ can be deformed to enclose the
line
 segment $[0, a_1]$, to match the one in
\cite{Ba 2015}.

\noindent Let us define the moments,
\begin{equation}
\begin{array}{lll}\displaystyle
\nu_{jk}^{(i)}:=\displaystyle\frac{1}{2\ii}\int_{\Gamma}
z^{j+k}\,\chi_{{\bf{n}}-{\bf e}_i}(z)\,\dd
        z=\displaystyle\frac{1}{2\ii}\int_{\Gamma}
z^{j+k}\,\widetilde{\chi}_{{\bf{n}}-{\bf e}_i}(z)\,\dd
        z,\vspace{0.3cm}\\\displaystyle
        \mu_{jk}:=\displaystyle\frac{1}{2\ii}\int_{\Gamma}
z^{j}\,\chi^\infty_{k}(z)\,\dd
        z=\int_\CC z^j\,\bar{z}^k\,\ee^{-|z|^2} |W(z)|^2\,\dd A(z).
\end{array}
\end{equation}
\begin{theorem}\label{thm3}
Let $n,\,l,\,\kappa,\,r$ and ${\bf
n}=(\underbrace{\kappa+1,\cdots,\kappa+1}_{r},\underbrace{\kappa,\cdots,\kappa}_{l-r})$
be given as in Theorem \ref{thm1}. For $\nu_{jk}^{(i)}$ and
$\mu_{jk}$ given above, set the $n$ by $n$ matrices of moments
$d_{n}$ and $D_{n}$ by
$$d_{n}=\left[ \setstretch{1.75}{
\begin{array}{c}\begin{matrix} \vdots\end{matrix}\\\hdashline
\begin{matrix}
\nu_{0,\,0}^{(j)}&\nu_{1,\,0}^{(j)}&\cdots&\nu_{n-1,\,0}^{(j)} \\
\vdots&\vdots&\vdots&\vdots\\
\nu_{0,\,n_j-1}^{(j)}&\nu_{1,\,n_j-1}^{(j)}&\cdots&\nu_{n-1,\,n_j-1}^{(j)}\end{matrix}
\\\hdashline
\begin{matrix}
\vdots
\end{matrix}
\end{array}}
\right],\quad D_{n}= \begin{bmatrix}
\mu_{0,\,0}&\mu_{1,\,0}&\cdots&\mu_{n-1,\,0} \\
\mu_{0,\,1}&\mu_{1,\,1}&\cdots&\mu_{n-1,\,1} \\
\vdots&\vdots&\ddots&\vdots\\
\mu_{0,\,n-1}&\mu_{1,\,n-1}&\cdots&\mu_{n-1,\,n-1}
\end{bmatrix},$$
where $$n_i = \begin{cases} \kappa+1 \quad
     \mbox{for}\quad 1\leq i\leq r; \\ \kappa   \quad
     \mbox{for}\quad  r+1\leq i\leq l. \end{cases}$$ Then there exists a unique constant matrix $A_n$ such that
$d_{n}={A_n} D_{n}. $ Moreover it satisfies
\begin{equation}\label{eqtau}
\begin{array}{lll}
\det
A_n&=&\displaystyle(-1)^{n(n-1)/2}\left(\prod_{i=1}^{l}\prod_{j=1}^{n_i-1}(c_i+j)^j\right)\prod_{
i<j}(\bar{a}_j-\bar{a}_i)^{n_in_j}\vspace{0.2cm}\\
&=&\displaystyle(-1)^{n(n-1)/2}\left(\prod_{i=1}^{l}\prod_{j=1}^{\kappa-1}(c_i+j)^j\right)\left(\prod_{i=1}^r(c_i+\kappa)^{\kappa}\right)\prod_{1\leq
i<j\leq\, l}(\bar{a}_j-\bar{a}_i)^{\kappa^2}\vspace{0.2cm}\\
&\times& \displaystyle\prod_{1\leq i<j\leq\,
r}(\bar{a}_j-\bar{a}_i)^{2\kappa+1}\prod_{j=r+1}^l\prod_{i=1}^r(\bar{a}_j-\bar{a}_i)^\kappa.\end{array}\end{equation}
\end{theorem}

Theorem 2 provides a way to study such planar orthogonal
polynomials, namely, by the nonlinear steepest descent analysis of
matrix Riemann--Hilbert problem, see \cite{Ba 2015, ku94 2015, ku103
2015, ku106 2015}. Theorem 3 suggests that the partition function of
the corresponding Coulomb Gas system (see \cite{Lee 2017} and the
reference therein) can be calculated using the  tau--function from
the Riemann-Hilbert problem \cite{be 2009}. Both directions are
currently in progress by the authors.

\section{Proof of Theorem \ref{thm1}}
\subsection{Area Integral via Contour Integral}

The following definitions will be useful.
\begin{equation}
\begin{split}
\chi_m(z)&:=W(z)\int_{0}^{\bar{z}}s^m\overline{W}(s)\, \ee^{- z s }
\dd s,
    \\
\chi_{m}^{\infty}(z)&:=W(z)\int_0^{\overline{z}\times\infty}
s^m\overline{W}(s)\, \ee^{-z s } \dd s.
    \end{split}
\end{equation}
Both are well defined if $\arg z\neq \arg a_j$ for all $j.$ They
satisfy the following lemma.

\begin{lemma}\label{lem-branch1}
Let ${\bf S}=\bigcup_{j=1}^{\,l}{\bf S}_j$ where ${\bf S}_j=\{a_jt:
0\leq t\leq 1\}.$ $\chi_{m}^{\infty}(z)-\chi_m(z)$ has continuous
extension in $\CC\setminus {\bf S}$ and, given $k>0,$ there exists
$C>0$ such that
\begin{equation}\label{equneq}
|z^k|\left|\chi_{m}^{\infty}(z)-\chi_m(z)\right|\leq
C\,\ee^{-(|z|-1)^2}\end{equation} for all $z$ such that $|z|>2.$
\end{lemma}
\begin{proof}
It is enough to check the continuity on
${\bf{B}}_1\setminus\{a_1\}.$ The piecewise analytic functions, $W$
and $\overline{W}$, satisfy the following jump conditions,
\begin{equation}\label{eq-branch1}
\begin{split}
    &W_+(z)=\ee^{-2\pi\ii c_j}W_-(z) , \quad  z\in {\bf B}_j,
    \\
    &\overline W_+(z)=\ee^{-2\pi\ii c_j}\overline W_-(z) , \quad  z\in {{\bf B}^*_j}.
        \end{split}
\end{equation}
Here the subscripts $\pm$ stand for the boundary values taken from
$\pm$ sides of ${\bf B}$;  we assign $\pm$ sides on each point of
${\bf B}\setminus\{a_1, a_2, \cdots,a_l\}$ and ${\bf
B}^*\setminus\{\bar{a}_1, \bar{a}_2, \cdots,\bar{a}_l\}$ in a
standard way, see Figure \ref{fig-branch-cuts}.

Let $p\in{\bf B}_1\setminus\{a_1\}.$ Note that when $z$ approaches
$p$ from $+$ side of ${\bf B}_1$, $\overline z$ approaches ${\bf
B}^{*}_1$ from $-$ side. Then we get
\begin{equation}
\begin{array}{lll}\displaystyle
[\chi_m^{\infty}(p)-\chi_m(p)]_+&=&[W(p)]_+\displaystyle
\int_{\bar{p}}^{\bar{p}\times\infty}
s^m\left[\overline{W}(s)\right]_-\,\ee^{-p s } \dd
s\vspace{0.2cm}\\\displaystyle &=&
    [W(p)]_- \int_{\bar{p}}^{\bar{p}\times\infty}
s^m\left[\overline{W}(s)\right]_+\,\ee^{-p s } \dd s
\vspace{0.2cm}\\\displaystyle &=&[\chi_m^{\infty}(p)-\chi_m(p)]_-,
\end{array}
\end{equation}
where we used \eqref{eq-branch1} at the second equality. This proves
the continuity statement. To prove the statement about the bound, we
use the elementary estimate that, given $k>0,$ there exists $C>0$
such that $$|z^k|\left|W(z)\right|\leq C\, \ee^{|z|}$$ for all
$z\in\CC$. Then, for some $C>0$ and $|z|>2,$ we get
\begin{equation}
\begin{array}{lll}\displaystyle
\left|z^k\left(\chi_m^{\infty}(z)-\chi_m(z)\right)\right|&=&\displaystyle\left|z^k
W(z) \int_{\bar{z}}^{\bar{z}\times\infty}
s^m\overline{W}(s)\,\ee^{-z s } \dd
s\right|\vspace{0.2cm}\\\displaystyle &\leq&\displaystyle
C\,\ee^{|z|} \int_{\bar{z}}^{\bar{z}\times\infty} \ee^{|s|}\,\ee^{-z
s } |\dd s|\vspace{0.2cm}\\ &\leq&C\,\ee^{|z|}\displaystyle
\left|\int_{|\bar{z}|}^{\infty} \ee^{x}\,\ee^{-|z|x } \dd
x\right|=C\,\frac{\ee^{-|z|^2+2|z|}}{|z|-1}\leq\widetilde{C}\,{\ee^{-(|z|-1)^2}}.
\end{array}
\end{equation}
\end{proof}

\begin{proposition}\label{prop-main} For an arbitrary polynomial $p(z)$ we have the following identity:
    \begin{equation}
        \int_\CC p(z)\,\bar{z}^m\,\ee^{-|z|^2} |W(z)|^2\,\dd A(z)=\frac{1}{2\ii}\int_{\Gamma} p(z)\,\chi_{m}^{\infty}(z)\,\dd
        z.
    \end{equation}
\end{proposition}
\begin{proof}
We apply Green's theorem to change the integral over $\CC$ to the
integral over a contour.  First we observe that
\begin{equation}
{\bar z}^m\,|W(z)|^2 \ee^{-|z|^2}=\frac{\partial
{\chi_m}(z)}{\partial \overline z} ,\quad z\in\CC\setminus{\bf B}.
\end{equation}
Therefore, defining $D_R:=\{z\,|\,|z|<R\},$ we get
\begin{equation}\label{eq-13}
\begin{split}
\int_\CC p(z)\,{\bar z}^m\,|W(z)|^2 \ee^{-|z|^2}\,\dd
A(z)&=\lim_{R\to\infty}\int_{D_R} p(z)\,{\bar z}^m\,|W(z)|^2
\ee^{-|z|^2}\,\dd A(z)\\&=\lim_{R\to\infty}\int_{D_R\setminus{\bf
B}} p(z)\, \frac{\partial {\chi_m}(z)}{\partial \bar z}\,\dd A(z)
\\
&= \lim_{R\to\infty}\frac1{2\ii}\bigg(\int_{\partial D_R}
p(z)\,{\chi_m}(z)\,\dd z+ \sum_{j=1}^m \int_{{\bf B}_j\cap D_R}
p(z)\big[{\chi_m}(z)\big]^+_-\dd z\bigg),
\end{split}
\end{equation}
where we use Green's theorem at the last equality.

Since ${\chi_m^{(\infty)}}(z)$ is analytic in $\CC\setminus
({\bf{S}}\cup{\bf{B}}),$ by deformation of contour we get the
identity
\begin{equation}
\int_{\Gamma} p(z)\,{\chi_m^{\infty}}(z)\,\dd z=\int_{\partial D_R}
p(z)\,{\chi_m^{\infty}}(z)\,\dd z+\sum_{j=1}^m \int_{{\bf B}_j\cap
D_R} p(z)\big[{\chi_m}(z)\big]^+_-\dd z
\end{equation}
Using this identity, the right hand side of \eqref{eq-13} becomes
\begin{equation}
\lim_{R\to\infty}\frac1{2\ii}\int_{\partial D_R}
p(z)\,\left({\chi_m}(z)-{\chi_m^{\infty}}(z)\right)\,\dd
z+\frac1{2\ii}\int_{\Gamma} p(z)\,{\chi_m^{\infty}}(z)\,\dd z
=\frac1{2\ii}\int_{\Gamma} p(z)\,{\chi_m^{\infty}}(z)\,\dd z,
\end{equation}
where the last equality holds because of \eqref{equneq} in Lemma
\ref{lem-branch1}.
 This proves Proposition \ref{prop-main}.
\end{proof}
\subsection{Several Lemmas}
\begin{definition}
All the vectors in this paper have only non-negative entries. For
two vectors, ${\bf k}$ and ${\bf s},$ we say ${\bf k}\geq {\bf s}$
if ${\bf k}-{\bf s}$ has only non-negative entries. If, in addition,
${\bf k}\neq{\bf s}$ then we say ${\bf k}>{\bf s}.$ The $j$th entry
of ${\bf k}$ is denoted by $[{\bf k}]_j.$ We define the length of a
vector by $|{\bf k}|=[{\bf k}]_1+\cdots+[{\bf k}]_l.$
\end{definition}

\begin{lemma}\label{lm7}
For any $n\geq 1$ we have
\begin{equation}\label{eq48}
\text{span}\,\{\chi_{j}^{\infty}:\,0\leq j<n\}=\text{span}\,
\{\chi_{{\bf k}}:\,|{\bf k}|\leq n\} .
\end{equation}
\end{lemma}
\begin{proof}
For $n=0,$ the lemma holds because $\chi_{0}^{\infty}(z)=\chi_{{\bf
0}}(z).$ Assume that the lemma holds for $n=n_0.$ If $|{\bf
k}|=n_0+1$ we get
\begin{equation}
\begin{array}{lll}\displaystyle
\chi_{\bf k}(z)-\chi_{n_0+1}^{\infty}(z)&=&\displaystyle
W(z)\int_0^{\bar{z}\times\infty}\prod_{j=1}^l(s-\bar{a}_j)^{k_j}\overline{W}(s)\ee^{-zs}\,\dd
s-W(z)\int_0^{\bar{z}\times\infty}s^{n_0+1}\overline{W}(s)\ee^{-zs}\,\dd
s\vspace{0.2cm}\\\displaystyle &=&\displaystyle
W(z)\int_0^{\bar{z}\times\infty}\{\text{polynomial in s of degree
}\leq n_0\}\times\overline{W}(s)\ee^{-zs}\,\dd s.
\end{array}
\end{equation}
Since the last term belongs to both spans in \eqref{eq48} for
$n=n_0,$ $\chi_{\bf k}$ belongs to the left span in \eqref{eq48}
with $n=n_0+1$ and $\chi_{n_0+1}^{\infty}$ belongs to the right span
in \eqref{eq48} with $n=n_0+1.$
\end{proof}

To prove $p_n=p_{{\bf n}}$, one may try to show that
\begin{equation}\label{eq-chi}
\text{span}\,\{\chi_{j}^{\infty}(z)|\,0\leq j<n\}=\text{span}\,
\{z^k\chi_{{\bf n}-{\bf{e}}_j}\,|\,0\leq k< [{\bf n}]_j, 1\leq j
\leq l\} .
\end{equation}
In fact, it is enough to show that the above equality up to
functions $\psi$ that satisfies $\langle p, \psi\rangle=0$ for all
polynomial $p.$ For example, we have $\langle p,\psi\rangle=0$ for
\begin{equation}
\psi(z)=W(z)\int_0^{\bar{a}_1}\prod_{j=1}^l(s-\bar{a}_j)^{k_j}\overline{W}(s)\ee^{-zs}\,\dd
s.
\end{equation}
Since $\psi$ is analytic in $\CC\setminus{\bf B}$ and, therefore,
the integration contour in $\int_{\Gamma}p(z)\psi(z)dz$ is
contractible to a point. This allows us to consider, instead of
$\chi_{\bf k}$ in \eqref{eq-chi},
$$\widetilde{\chi}_{\bf k}:=\chi_{\bf
k}-W(z)\int_0^{\bar{a}_1}\prod_{j=1}^l(s-\bar{a}_j)^{k_j}\overline{W}(s)\ee^{-zs}\,\dd
s.$$ As a result, using Lemma \ref{lm7}, the proof of Theorem
\ref{thm1} is reduced to proving the following Proposition.
\begin{proposition}\label{prop2}
For any $n\geq 1$ and $l\geq 1$ let ${\bf n}$ be given as in Theorem
\ref{thm1}. Then the following holds.
\begin{equation}
\text{span}\,\{\widetilde{\chi}_{\bf k}(z):\,|{\bf
k}|<n\}=\text{span}\, \{z^k\widetilde{\chi}_{{\bf
n}-{\bf{e}}_j}(z)\,\big|\,0\leq k< [{\bf n}]_j, 1\leq j \leq l\} .
\end{equation}
\end{proposition}
The proof of this proposition will be in the next subsection. The
following Lemma is why it is useful to use $\widetilde{\chi}_{\bf
k}$ instead of ${\chi}_{\bf k}.$

\begin{lemma}\label{lm1}
\begin{equation}\label{eq-z}
z\widetilde{\chi}_{\bf{k}}(z)=\sum_{j=1}^l(c_j+k_j)\widetilde{\chi}_{{\bf{k}}-{\bf{e}}_j}(z).
\end{equation}
\end{lemma}
\begin{proof}
Taking the integral of the total derivative as following, we have
\begin{equation}\nonumber
\begin{array}{lll}\displaystyle
0&=& \displaystyle
W(z)\int_{\bar{a}_1}^{\overline{z}\times\infty}\partial_s\left[\prod_{j=1}^l(s-\bar{a}_j)^{c_j+k_j}\ee^{-zs}\right]\,\dd
s\vspace{0.2cm}\\\displaystyle &=&\displaystyle
W(z)\int_{\bar{a}_1}^{\overline{z}\times\infty}\left(\sum_{j=1}^l\frac{c_j+k_j}{s-\bar{a}_j}
-z\right)\prod_{j=1}^l(s-\bar{a}_j)^{c_j+k_j}\ee^{-zs}\,\dd
s\vspace{0.2cm}\\\displaystyle &=&\displaystyle
\sum_{j=1}^l(c_j+k_j)\widetilde{\chi}_{{\bf{k}}-e_j}(z)-z\widetilde{\chi}_{{\bf{k}}}(z).
\end{array}
\end{equation}
\end{proof}

\begin{cor}\label{lm3}
Let ${\bf{k}}=(k_1,k_2,\cdots,k_l)$ and $s\leq \min\{k_j\}_{j=1}^l$
be a positive integer. Then $z^s\widetilde{\chi}_{\bf{k}}(z)$ can be
represented as a linear combination of
$\{\widetilde{\chi}_{{\bf{k}}-{\bf{s}}}(z)\big|\, |{\bf{s}}|=s\}.$
Furthermore, the coefficient of
$\widetilde{\chi}_{{\bf{k}}-s{\bf{e}}_m}(z)$ is nonzero for all
$1\leq m\leq l$.
\end{cor}
\begin{proof}
From Lemma \ref{lm1}, the corollary is true when $s=1$. Assume, for
some $1\leq s< \min\{k_j\}_{j=1}^l,$ that
$z^s\widetilde{\chi}_{\bf{k}}(z)$ is a linear combination of
$\widetilde{\chi}_{{\bf{k}}-{\bf{s}}}(z)$ for $|{\bf{s}}|=s$ and the
coefficient of
$\{\widetilde{\chi}_{{\bf{k}}-s{\bf{e}}_m}(z)\}_{m=1}^l$ are all
non-vanishing.

Then $z^{s+1}\widetilde{\chi}_{\bf{k}}(z)$ is a linear combination
of $z\widetilde{\chi}_{{\bf{k}}-{\bf{s}}}(z)$ and, therefore, of
$\widetilde{\chi}_{{\bf{k}}-{\bf{s}}-{\bf e}_m}(z)$ with $|{\bf
s}|=s$ and $1\leq m\leq l.$ Since the term
$\widetilde{\chi}_{{\bf{k}}-(s+1){\bf{e}}_m}(z)$ comes only from
$z\widetilde{\chi}_{{\bf{k}}-s{\bf{e}}_m}(z)$ and since the
coefficient of $\widetilde{\chi}_{{\bf{k}}-s{\bf{e}}_m}(z)$ is
non-zero, the coefficient of
$\widetilde{\chi}_{{\bf{k}}-(s+1){\bf{e}}_m}(z)$ is non-zero. Note
that all the coefficients in the right hand side of \eqref{eq-z} are
non-zero. By induction, this ends the proof.
\end{proof}

\begin{lemma}\label{lm2}
For $n\neq m,$ we have
\begin{equation}\label{eq43}
\widetilde{\chi}_{{\bf{k}}+{\bf{e}}_n}(z)-\widetilde{\chi}_{{\bf{k}}+{\bf{e}}_m}(z)+\left(\bar{a}_n-\bar{a}_m\right)\widetilde{\chi}_{{\bf{k}}}(z)=0.
\end{equation}
\end{lemma}
\begin{proof}
Since
$$\left(s-\bar{a}_n\right)-\left(s-\bar{a}_m\right)+\left(\bar{a}_n-\bar{a}_m\right)=0,$$
we obtain,
$$0=W(z)\int_{\bar{a}_1}^{\bar{z}\times\infty}\left[\left(s-\bar{a}_n\right)-\left(s-\bar{a}_m\right)+
\left(\bar{a}_n-\bar{a}_m\right)\right]\prod_{j=1}^l(s-\bar{a}_j)^{c_j+k_j}\ee^{-zs}\dd
s.$$ By the definition of $\widetilde{\chi}_{\bf{k}}(z),$
\eqref{eq43} holds.
\end{proof}
\subsection{Proof of Proposition \ref{prop2}}
By Corollary \ref{lm3}, we get $\supset.$ To prove $\subset,$ we
note that any vector ${\bf k}$ can be uniquely represented as
$${\bf k}={\bf n}+{\bf m}-{\bf s},$$ where $[{\bf m}]_j[{\bf
s}]_j=0,$ i.e., ${\bf m}$ and ${\bf s}$ cannot be both non-vanishing
in any of the entries. It is then enough to show the following
claim.

\bigskip
\noindent{\bf{Claim:}} For all ${\bf s}\leq{\bf n}$ and ${\bf m}$
satisfying $|{\bf n}+{\bf m}-{\bf s}|<n,$ $$\widetilde{\chi}_{{\bf
n}+{\bf m}-{\bf s}}\in \text{span}\, \{z^k\widetilde{\chi}_{{\bf
n}-{\bf{e}}_j}(z)\,\big|\,0\leq k< [{\bf n}]_j, 1\leq j \leq l\}.$$
\bigskip

We prove this claim in two steps.

\noindent\textbf{Step 1:} For all ${\bf 0}<{\bf s}\leq {\bf n},$
$\widetilde{\chi}_{{\bf n}-{\bf s}}\in\text{span}\,
\{z^k\widetilde{\chi}_{{\bf n}-{\bf{e}}_j}(z)\,\big|\,0\leq k< [{\bf
n}]_j, 1\leq j \leq l\}.$ If $|{\bf s}|=1$ then the inclusion is
immediate. Let the inclusion holds for $|{\bf s}|\leq m-1$ for some
$m<n.$ (If $m\geq n$ then the proof is done.) Below we claim that
the inclusion holds for $|{\bf s}|=m,$ which proves Step 1 by
induction.
\begin{enumerate}
\item\label{i1} If ${\bf s}$ has more than one non-zero entries, i.e., $[{\bf s}]_i\neq
0$ and $[{\bf s}]_j\neq 0,$ $$\widetilde{\chi}_{{\bf n}-{\bf
s}}(z)=\frac{1}{\bar{a}_i-\bar{a}_j}\left(\widetilde{\chi}_{{\bf
n}-{\bf s}+{\bf{e}}_j}(z)-\widetilde{\chi}_{{\bf n}-{\bf
s}+{\bf{e}}_i}(z)\right).$$ The left hand side belongs to the span
in Claim since the right hand side does by assumption.
\item If ${\bf s}$ has exactly one non-zero entry, i.e., ${\bf s}=m{\bf
e}_j$ for some $j.$ From ${\bf s}<{\bf n}$ we have $m\leq [{\bf
n}]_j.$ Since $z^{m-1}\widetilde{\chi}_{{\bf n}-{\bf{e}}_j}(z)$ is a
linear combination of $\{\widetilde{\chi}_{{\bf n}-{\widetilde{\bf
s}}}: |{\widetilde{\bf s}}|=m\}$ where the term
$\widetilde{\chi}_{{\bf n}-m{\bf{e}}_j}$ appears with non-zero
coefficient (see Corollary \ref{lm3}), and since all the other terms
in the linear combination belongs to the span by item \ref{i1},
$\widetilde{\chi}_{{\bf n}-m{\bf{e}}_j}$ also belongs to the span in
Claim.
\end{enumerate}

\noindent \textbf{Step 2:} Step 1 showed Claim for $|{\bf m}|=0.$
Assume that Claim is true when $|{\bf m}|\leq k-1.$ We will show
that Claim holds when $|{\bf m}|\leq k$, i.e.
$\widetilde{\chi}_{{\bf n}+{\bf m}-{\bf s}}$ belongs to the span in
Claim for $|{\bf m}|=k.$ Let ${\bf m}$ satisfy $|{\bf m}|=k\geq 1.$
There exists $j$ such that $[{\bf m}]_j>0.$ Then
$\widetilde{\chi}_{{\bf n}+({\bf m}-{\bf{e}}_j)-{\bf s}}$ belongs to
the span in the claim by the assumption. Since $|{{\bf n}+({\bf
m}-{\bf{e}}_j)-{\bf s}}|<n-1$ we have $|{\bf s}|>0$ and there exists
$i\neq j$ such that $[{\bf s}]_i>0.$ Then $\widetilde{\chi}_{{\bf
n}+({\bf m}-{\bf{e}}_j)-({\bf s}-{\bf{e}}_i)}$ also belongs to the
span by the assumption. Since, by Lemma \ref{lm2}, we have
\begin{equation}\nonumber
\widetilde{\chi}_{{\bf n}+{\bf m}-{\bf s}}=\widetilde{\chi}_{{\bf
n}+({\bf m}-{\bf{e}}_j)-({\bf
s}-{\bf{e}}_i)}+(\bar{a}_i-\bar{a}_j)\widetilde{\chi}_{{\bf n}+({\bf
m}-{\bf{e}}_j)-{\bf s}},\end{equation} the left hand side belongs to
the span. This ends the proof of Proposition \ref{prop2} and Theorem
\ref{thm1}.

\section{Proof of Theorem \ref{thm3}}
Since $\det D_{n}=\prod_{j=0}^{n-1}h_j>0$ where $h_j$ is defined in
\eqref{eq1}, $D_{n}$ is an invertible matrix and this proves the
existence and the uniqueness of $A_n.$ In the remainder of the
proof, we will construct $A_n$ using induction. Let us consider the
$j$th column of $d_{n},$
$$\setstretch{1.5}{\begin{bmatrix}\displaystyle
\nu_{j,\,0}^{(1)}\vspace{0.1cm}\\
\nu_{j,\,1}^{(1)}\\
\vdots\\
\nu_{j,\,n_1-1}^{(1)} \\
\vdots\\
\nu_{j,\,0}^{(l)}\vspace{0.1cm}\\
\nu_{j,\,1}^{(l)}\\
\vdots\\
\nu_{j,\,n_l-1}^{(l)}
\end{bmatrix}}=\displaystyle\frac{1}{2\ii}\int_{\Gamma}z^j\,
V_{\bf n}(z)\,\dd
        z,\quad\text{where}\,\,\, V_{\bf n}=V_{\bf n}(z)=\displaystyle
\setstretch{1.5}{\begin{bmatrix}
\chi_{{\bf{n}}-{\bf e}_1}\\
z\chi_{{\bf{n}}-{\bf e}_1}\\
\vdots\\
z^{n_1-1}\,\chi_{{\bf{n}}-{\bf e}_1} \\
\vdots\\
\chi_{{\bf{n}}-{\bf e}_l}\\
z\chi_{{\bf{n}}-{\bf e}_l}\\
\vdots\\
z^{n_l-1}\,\chi_{{\bf{n}}-{\bf e}_l} \\
\end{bmatrix}}.$$ We will find a constant $(n+1)$ by $(n+1)$ matrix $B_n$ such
that, for all $z,$
$$B_n  V_{{\bf{n}}+{\bf
e}_{r+1}}(z)=\left[\setstretch{1.2}{
\begin{array}{c}
\chi_{{\bf{n}}}(z)\\\hdashline V_{\bf n}(z)
\end{array}}\right].$$
This means that
$$B_n  d_{n+1}=\left[\setstretch{2}{
\begin{array}{c:c}
  \displaystyle \begin{matrix}\nu_{0,\,0}&\nu_{1,\,0}&\cdots&\nu_{n-1,\,0}\end{matrix} & \nu_{n,\,0}
  \vspace{0.1cm}\\\hdashline
   d_{n}=A_n D_n & \begin{matrix}\displaystyle
\nu_{n,\,0}^{(1)}\\
\vdots\\
\nu_{n,\,n_1-1}^{(1)} \\
\vdots\\
\nu_{n,\,0}^{(l)}\\
\vdots\\
\nu_{n,\,n_l-1}^{(l)}
\end{matrix} \\
\end{array}}\right],$$
where $\nu_{j,\,0}$ is given by
$\nu_{j,\,0}=\displaystyle\frac{1}{2\ii}\int_{\Gamma}
z^{j}\,\chi_{{\bf{n}}}(z)\,\dd z.$ The matrix $B_n$ can be obtained
by three successive linear transformations on $V_{{\bf{n}}+{\bf
e}_{r+1}}$ that we describe below.
\begin{equation}\nonumber
\left[\setstretch{1.1}{\,\begin{matrix} \chi_{{\bf{n}}+{\bf
e}_{r+1}-{\bf e}_1} \\
z\chi_{{\bf{n}}+{\bf e}_{r+1}-{\bf e}_1}\\
\vdots\\
z^\kappa\chi_{{\bf{n}}+{\bf
e}_{r+1}-{\bf e}_1}\vspace{0.1cm}\\
\hdashline \vdots\\\hdashline \chi_{{\bf{n}}+{\bf
e}_{r+1}-{\bf e}_{r}} \\
z\chi_{{\bf{n}}+{\bf e}_{r+1}-{\bf e}_{r}}\\
\vdots\\
z^\kappa\chi_{{\bf{n}}+{\bf e}_{r+1}-{\bf
e}_{r}}\vspace{0.1cm}\\\hdashline
 \chi_{{\bf{n}}} \\
z\chi_{{\bf{n}}} \\
\vdots\\ z^\kappa\chi_{{\bf{n}}}
\vspace{0.1cm}\\\hdashline\chi_{{\bf{n}}+{\bf
e}_{r+1}-{\bf e}_{r+2}} \\
z\chi_{{\bf{n}}+{\bf e}_{r+1}-{\bf e}_{r+2}}\\
\vdots\\
z^{\kappa-1}\chi_{{\bf{n}}+{\bf e}_{r+1}-{\bf
e}_{r+2}}\vspace{0.1cm}\\\hdashline \vdots\\\hdashline
\chi_{{\bf{n}}+{\bf
e}_{r+1}-{\bf e}_l} \\
z\chi_{{\bf{n}}+{\bf e}_{r+1}-{\bf e}_l}\\
\vdots\\
z^{\kappa-1}\chi_{{\bf{n}}+{\bf e}_{r+1}-{\bf e}_l}
\end{matrix}\,\,}\right]\,\xrightarrow{(A)}\,\left[\setstretch{1.1}{\,\begin{matrix}
\chi_{{\bf{n}}-{\bf e}_1} \\
z\chi_{{\bf{n}}-{\bf e}_1}\\
\vdots\\
z^\kappa\chi_{{\bf{n}}-{\bf e}_1}\vspace{0.1cm}\\\hdashline
\vdots\\\hdashline
\chi_{{\bf{n}}-{\bf e}_{r}} \\
z\chi_{{\bf{n}}-{\bf e}_{r}}\\
\vdots\\
z^\kappa\chi_{{\bf{n}}-{\bf e}_{r}}\vspace{0.1cm}\\\hdashline
 \chi_{{\bf{n}}} \\
z\chi_{{\bf{n}}}\\
\vdots\\ z^\kappa\chi_{{\bf{n}}} \vspace{0.1cm}\\\hdashline\chi_{{\bf{n}}-{\bf e}_{r+2}} \\
z\chi_{{\bf{n}}-{\bf e}_{r+2}}\\
\vdots\\
z^{\kappa-1}\chi_{{\bf{n}}-{\bf e}_{r+2}}\vspace{0.1cm}\\\hdashline
\vdots\\\hdashline
\chi_{{\bf{n}}-{\bf e}_l} \\
z\chi_{{\bf{n}}-{\bf e}_l}\\
\vdots\\
z^{\kappa-1}\chi_{{\bf{n}}-{\bf e}_l}
\end{matrix}\,\,}\right]
\,\xrightarrow{(B)}\,\left[\setstretch{1.1}{\,\begin{matrix}
\chi_{{\bf{n}}-{\bf e}_1} \\
z\chi_{{\bf{n}}-{\bf e}_1}\\
\vdots\\
z^\kappa\chi_{{\bf{n}}-{\bf e}_1}\vspace{0.1cm}\\\hdashline
\vdots\\\hdashline
\chi_{{\bf{n}}-{\bf e}_{r}} \\
z\chi_{{\bf{n}}-{\bf e}_{r}}\\
\vdots\\
z^\kappa\chi_{{\bf{n}}-{\bf e}_{r}}\vspace{0.1cm}\\\hdashline
 \chi_{{\bf{n}}} \\
\chi_{{\bf{n}}-{\bf
e}_{r+1}} \\
\vdots\\ z^{\kappa-1}\chi_{{\bf{n}}-{\bf
e}_{r+1}} \vspace{0.1cm}\\\hdashline\chi_{{\bf{n}}-{\bf e}_{r+2}} \\
z\chi_{{\bf{n}}-{\bf e}_{r+2}}\\
\vdots\\
z^{\kappa-1}\chi_{{\bf{n}}-{\bf e}_{r+2}}\vspace{0.1cm}\\\hdashline
\vdots\\\hdashline
\chi_{{\bf{n}}-{\bf e}_l} \\
z\chi_{{\bf{n}}-{\bf e}_l}\\
\vdots\\
z^{\kappa-1}\chi_{{\bf{n}}-{\bf e}_l}
\end{matrix}\,\,}\right]
\,\xrightarrow{(C)}\,\left[\setstretch{1.1}{\,\begin{matrix}
\chi_{{\bf{n}}} \\
\chi_{{\bf{n}}-{\bf e}_1} \\
z\chi_{{\bf{n}}-{\bf e}_1}\\
\vdots\\
z^\kappa\chi_{{\bf{n}}-{\bf e}_1}\vspace{0.1cm}\\\hdashline
\vdots\\\hdashline
\chi_{{\bf{n}}-{\bf e}_r} \\
z\chi_{{\bf{n}}-{\bf e}_r}\\
\vdots\\
z^\kappa\chi_{{\bf{n}}-{\bf e}_r}\vspace{0.1cm}\\\hdashline
\chi_{{\bf{n}}-{\bf
e}_{r+1}} \\
\vdots\\ z^{\kappa-1}\chi_{{\bf{n}}-{\bf
e}_{r+1}} \vspace{0.1cm}\\\hdashline\chi_{{\bf{n}}-{\bf e}_{r+2}} \\
z\chi_{{\bf{n}}-{\bf e}_{r+2}}\\
\vdots\\
z^{\kappa-1}\chi_{{\bf{n}}-{\bf e}_{r+2}}\vspace{0.1cm}\\\hdashline
\vdots\\\hdashline
\chi_{{\bf{n}}-{\bf e}_l} \\
z\chi_{{\bf{n}}-{\bf e}_l}\\
\vdots\\
z^{\kappa-1}\chi_{{\bf{n}}-{\bf e}_l}
\end{matrix}\,\,}\right],
\end{equation}
Above, each arrow means the linear transformation given by
\begin{equation}\nonumber
\begin{array}{lll}
B_{n}^{(1)}\,\text{LHS of (A)}\,&=&\,\text{RHS of
(A)},\vspace{0.3cm}
\\
B_{n}^{(2)}\,\text{LHS of (B)}\,&=&\,\text{RHS of
(B)},\vspace{0.3cm}
\\
B_{n}^{(3)}\,\text{LHS of (C)}\,&=&\,\text{RHS of (C)},
\end{array}
\end{equation}
\begin{equation}\nonumber
\begin{array}{lll}
B_{n}^{(1)}&=&\left[{\setstretch{2} \begin{array}{c:c:c}
\begin{matrix}
\displaystyle\frac{I_{\kappa+1}}{\bar{a}_1-\bar{a}_{r+1}}&\cdots&{\bf 0}\\
\vdots&\ddots&\vdots\\
{\bf
0}&\cdots&\displaystyle\frac{I_{\kappa+1}}{\bar{a}_r-\bar{a}_{r+1}}
\end{matrix}
&\begin{matrix}\displaystyle-\frac{I_{\kappa+1}}{\bar{a}_1-\bar{a}_{r+1}}\\\vdots\\
-
\displaystyle\frac{I_{\kappa+1}}{\bar{a}_r-\bar{a}_{r+1}}\end{matrix}
&\begin{matrix}{\bf 0}\end{matrix}\vspace{0.1cm}\\\hdashline
\begin{matrix}
{\bf 0}
\end{matrix}
&\begin{matrix}I_{\kappa+1}\end{matrix} &\begin{matrix}{\bf
0}\end{matrix}\vspace{0.1cm}\\\hdashline
\begin{matrix}{\bf 0}\end{matrix}
&\begin{matrix}\displaystyle-\frac{I_{\kappa}}{\bar{a}_{r+2}-\bar{a}_{r+1}}\\\vdots\\
-
\displaystyle\frac{I_{\kappa}}{\bar{a}_{l}-\bar{a}_{r+1}}\end{matrix}
&\begin{matrix}
\displaystyle\frac{I_{\kappa}}{\bar{a}_{r+2}-\bar{a}_{r+1}}&\cdots&{\bf 0}\\
\vdots&\ddots&\vdots\\
{\bf
0}&\cdots&\displaystyle\frac{I_{\kappa}}{\bar{a}_l-\bar{a}_{r+1}}
\end{matrix}\end{array}}\right],\vspace{0.3cm}\\
B_{n}^{(2)}&=&\left[ \setstretch{2}{\begin{array}{c:c:c}
\begin{matrix}
I_{(\kappa+1)r+1}
\end{matrix}
&\begin{matrix}{\bf 0}\end{matrix} &\begin{matrix}{\bf
0}\end{matrix}\vspace{0.1cm}\\\hdashline
\begin{matrix}
\displaystyle-\frac{c_1+\kappa+1}{c_{r+1}+\kappa}I_{\kappa}&{\bf
0}_{\kappa\times1}&\cdots&\displaystyle-\frac{c_r+\kappa+1}{c_{r+1}+\kappa}I_{\kappa}&{\bf
0}_{\kappa\times2}
\end{matrix}
&\begin{matrix}\displaystyle\frac{I_{\kappa}}{c_{r+1}+\kappa}\end{matrix}
&\begin{matrix}
\displaystyle-\frac{c_{r+2}+\kappa}{c_{r+1}+\kappa}I_\kappa&\cdots&\displaystyle-\frac{c_l+\kappa}{c_{r+1}+\kappa}I_\kappa\end{matrix}\vspace{0.1cm}\\\hdashline
\begin{matrix}{\bf 0}\end{matrix}
&\begin{matrix}{\bf 0}\end{matrix} &\begin{matrix} \displaystyle
I_{\kappa(l-r-1)}
\end{matrix}\end{array}}\right],\vspace{0.3cm}\\
B_{n}^{(3)}&=&\left[ \setstretch{1.75}{\begin{array}{c:c:c}
\begin{matrix}
{\bf 0}
\end{matrix}
&\begin{matrix}1\end{matrix} &\begin{matrix}{\bf
0}\end{matrix}\\\hdashline
\begin{matrix}
I_{(\kappa+1)r}
\end{matrix}
&\begin{matrix}{\bf 0}_{(\kappa+1)\times1}\end{matrix}
&\begin{matrix}{\bf 0}\end{matrix}\vspace{0.1cm}\\\hdashline
\begin{matrix}{\bf 0}\end{matrix}
&\begin{matrix}{\bf 0}\end{matrix} &\begin{matrix} I_{\kappa(l-r)}
\end{matrix}\end{array}}\right],
\end{array}
\end{equation}
where $I_m$ is the $m$ by $m$ identity matrix and ${\bf 0}_{j\times
k}$ is the zero matrix of size $j$ by $k$.
 We used Lemma \ref{lm2} in the transformation $(A)$ and
Lemma \ref{lm1} in $(B).$ This gives $B_n=B_{n}^{(3)} B_{n}^{(2)}
B_{n}^{(1)}.$

Using $d_{n}=A_n D_{n}$ we obtain that
\begin{equation}\label{eqA}
B_n d_{n+1}=B_n A_{n+1} D_{n+1}=\left[\setstretch{1.75}{
\begin{array}{c:c}
   \begin{matrix}C_0&\cdots&C_{n-1}\end{matrix} & 1 \vspace{0.1cm}\\\hdashline
   A_{n} & {\bf 0} \\
\end{array}}\right]D_{n+1}.
\end{equation} The identity at the first row is obtained by
$$\nu_{j,\,0}=\displaystyle\frac{1}{2\ii}\int_{\Gamma}
z^{j}\,\chi_{{\bf{n}}}(z)\,\dd
        z=\displaystyle\frac{1}{2\ii}\int_{\Gamma}
z^{j}\,\sum_{k=0}^n C_k\chi^\infty_{k}(z)\,\dd
        z=\sum_{k=0}^n C_k\mu_{jk},$$ where $C_k$ is given by
        $\displaystyle\prod_{i=1}^l(z-\bar{a}_i)^{n_j}=\sum_{k=0}^{n}C_k z^k.$
 We also used that the upper $n$ by $n$ diagonal submatrix of
$D_{n+1}$ is $D_{n}.$

Taking the determinant of \eqref{eqA} and using $B_n=B_{n}^{(3)}
B_{n}^{(2)} B_{n}^{(1)},$ we get
\begin{equation}\label{equB}
\begin{array}{lll}\displaystyle&&\det A_{n+1}=(-1)^{(n+2)}\left(\det B_{n}^{(1)}\det B_{n}^{(2)}\det B_{n}^{(3)}\right)^{-1}\det A_{n}\vspace{0.2cm}\\
&=&\displaystyle(-1)^{(n+2)+\sum_{i\leq r}n_i}\displaystyle
\left(\prod_{i<r+1}\left(\bar{a}_i-\bar{a}_{r+1}\right)^{n_i}\right)
\left(\prod_{j>r+1}(\bar{a}_{r+1}-\bar{a}_{j})^{n_j}\right)\left(c_{r+1}+\kappa\right)^{\kappa}\det
A_{n}.
\end{array}\end{equation}

Now we prove \eqref{eqtau} by induction. When
${\bf{n}}=(1,0\cdots,0)$ (i.e. $\kappa=0$ and $r=1$), by the
definition of $\nu_{jk}^{(i)}$ and $\mu_{jk},$ we observe
$\nu_{0,\,0}^{(1)}=\mu_{0,\,0}.$ This proves $d_1=D_1$ with $\det
A_1=1.$ If \eqref{eqtau} holds up to $n\leq N$ then \eqref{eqtau}
holds for $n=N+1$ by \eqref{equB}. Remember that if ${\bf
n}(N,l)=(n_1,\cdots, n_l)$ and $N=\kappa l +r$ then ${\bf
n}(N+1,l)=(n_1,\cdots,n_{r+1}+1,\cdots, n_l),$ increasing only the
$(r+1)$th entry by one. This ends the proof of Theorem \ref{thm3}.

\end{document}